\documentclass[letterpaper,11pt]{article}

\usepackage{booktabs} 

\usepackage[utf8]{inputenc}
\usepackage{amssymb}
\usepackage[margin=1in]{geometry}
\usepackage{amsthm}
\usepackage{thmtools}
\usepackage{thm-restate}

\usepackage{amsmath}
\usepackage{mathtools}
\usepackage{bbm}
 \usepackage{amsfonts}

\usepackage{dsfont}
\usepackage{graphicx}
\usepackage{multirow}
\usepackage{soul}
\usepackage{color}
\usepackage{xcolor}
\usepackage[
backend=biber,
style=alphabetic,
sorting=ynt,
maxbibnames=15
]{biblatex}
\bibliography{refs}

\usepackage{hyperref}
\usepackage{graphicx}
\usepackage{subfigure}
\hypersetup{%
   breaklinks,%
   colorlinks=true,%
   linkcolor=black,%
   urlcolor=black,%
   citecolor=[rgb]{0,0,0.45}
}

\theoremstyle{plain}
\newtheorem{theorem}{Theorem}[section]
\newtheorem{proposition}[theorem]{Proposition}
\newtheorem{lemma}[theorem]{Lemma}

\newtheorem{definition}[theorem]{Definition}

\newtheorem{remark}[theorem]{Remark}
\theoremstyle{definition}

\usepackage{color-edits}[showdeletions]

\usepackage{authblk}

\title{
A Theoretical Model for Grit in Pursuing Ambitious Ends}
\author[1]{Avrim Blum}
\author[2]{Emily Diana}
\author[1]{Kavya Ravichandran}
\author[3]{Alexander Tolbert}

\affil[1]{Toyota Technological Institute at Chicago, USA}
\affil[2]{Carnegie Mellon University, Tepper School of Business, USA}
\affil[3]{Emory University, Department of Quantitative Theory and Methods, USA}

\date{}

\begin{document}

\maketitle

\begin{abstract}
Ambition and risk-taking have been heralded as important ways for marginalized communities to get out of cycles of poverty. 
As a result, educational messaging often encourages individuals 
to strengthen their personal resolve and develop characteristics such as discipline and grit to succeed in ambitious ends. 
However, recent work in philosophy and sociology highlights that this messaging often does more harm than good for students in these situations. We study similar questions using a different epistemic approach and in simple theoretical models -- we provide a quantitative model of decision-making between stable and risky choices in the {\em improving multi-armed bandits} framework. We use this model to first study how individuals' ``strategies'' are affected by their level of grittiness and how this affects their accrued rewards. Then, we study the impact of various interventions, such as increasing grit or providing a financial safety net. Our investigation of rational decision making involves two different formal models of rationality, the competitive ratio between the accrued reward and the optimal reward and Bayesian quantification of uncertainty.
\end{abstract}

\section{Introduction}

Scholars across various fields have long been interested in understanding how humans make decisions between immediate and long-term rewards in the face of uncertainty. Costs associated with a given action can also greatly influence how agents make these decisions, and so it is of interest to understand how to encourage exploratory and ambitious behavior, particularly in groups who have historically faced lack of access to such opportunities. One important factor that has received much attention in recent years is {\em grit}, with researchers across various fields studying the role resilience and optimism play not only in an individual's success but also in lifting disadvantaged communities out of bad circumstances. \par

While it is difficult to give a single, unifying definition of grit, philosophers Jennifer Morton and Sarah Paul provide a thesis that serves as our guiding qualitative description. They highlight that grit is rational (distinguishing it from delusional optimism) and that it is an outcome of {\em beliefs} an agent holds about their circumstances. In more detail:

\begin{quote}
    ``Grit is not simply the ability to withstand the pain of effort and setbacks, or to resist the siren song of easier rewards; it is a trait or capacity that consists partly in a kind of epistemic resilience. \par

    This is a descriptive rather than a normative claim, and it has not gone unnoticed by psychologists who study perseverance. Angela Duckworth emphasizes the relevance of hope in underwriting the capacity for grit, where hope is defined as the expectation that one’s efforts will pay off. And Martin Seligman touts the importance of optimism, which involves a distinctive style of explaining to oneself why good and bad events happen.'' \cite{morton2019grit}
\end{quote}

Currently, prevailing guidance in educational settings focuses heavily on personal characteristics such as grit, discipline, and resilience as the way to succeed in ambitious ends \cite{Gibbon_seligman_2020}, often even at the protest of the scholars whose work was used to justify this perspective \cite{kamenetz_key_2015}. Recent work studies what happens when students from disadvantaged backgrounds follow those lessons in attempting to succeed at ambitious long-term ends such as college; these works find that simply pushing grit can often have negative effects \cite{morton2019grit, wooten_precarious_2022}. What, then, distinguishes productive grit from delusional optimism? \par

In this work, we try to isolate where additional grit helps and where it hurts via a simple quantitative model. We also investigate how outcomes are influenced by having financial supplements. Our goal is to study decision-making dynamics in a controlled, quantitatively-defined setting. Our work (1) provides further understanding of the relationship between grit and financial support in succeeding when pursuing ambitious outcomes and, while doing so, (2) introduces a simple two-armed bandit theoretical model that shows promise as a formalization of a decision-making problem that juxtaposes stable reward against ambition.

\subsection{Our Approach To Studying Grit}

Morton and Paul emphasize that grit leads agents to {\em rationally} stick with an option that others are not willing to stick with. We consider two formal models of rationality in decision-making. First, we study agents who maximize their competitive ratio. Measuring the ratio between achieved outcome and best possible outcome in hindsight, the competitive ratio being maximized implies that the agent has minimal multiplicative ``regret'' about the past. Thus, this model of rationality is a backward-looking model. Secondly, we study agents who take a Bayesian perspective to uncertainty quantification. In this view, an agent has an explicit prior probability distribution over the possible outcomes and updates their posterior each time they receive new information. In constructing a prior, the agent explicitly quantifies what they think the future holds, making this a forward-looking model of rationality. In this paper, we will compare and contrast the aspects of grit we are able to study in each of these formal models of rationality and see how the conclusions we draw from each on similar formal models relate.

In order to understand the impact of grit, we study two things -- first, we consider the effect of grit on the ``policy'' or ``strategy'' that the agent follows, i.e., the actions they take; second, we study what the impact of that ``policy'' is on the reward they achieve. This modular breakdown will become especially useful when we consider the impact of a trust fund on the actions of an agent, allowing us to disentangle the effect of grit as a trait and other interventions that lead agents to behave similarly to those who are gritty.

\subsection{Our Contributions}

In this work, we propose studying grit in the {\em improving multi-armed bandits} framework. We develop an instance that allows us to satisfactorily investigate the impact of grit as a characteristic the agent has, and we show that the instance we propose is the simplest instance in which the strategy is non-trivial. In order to understand strategies resulting from grit, which as discussed above is rational, in this model, we must formalize what we mean by ``rational,'' and we do so by appealing to two notions of rationality well-studied in computer science. In particular, we first consider the competitive ratio, a standard notion in the analysis of online algorithms, which is optimized when a strategy minimizes multiplicative regret in hindsight; we also study a Bayesian notion of rationality in which an agent has a prior that helps quantify their uncertainty about the future, which they update to a posterior based on evidence from their environment. Between these two notions, we study how gritty behavior is reflected in both forward- and backward- looking formalizations of rationality. \par

With the model and notions of rationality in hand, we study how grit affects an agent's strategy and how that in turn affects their outcomes. When grit is related to the optimistic outlook of an agent, we can quantitatively derive cases in which grit helps and hurts the agent, showing that an excess of grit harms the agent by causing them to net less reward than their less-gritty counterparts. When grit is reflected in how willing an agent is to tolerate discomfort, we conclude that though agents who require comfort can minimize their multiplicative regret in hindsight pretty easily, agents who can tolerate more discomfort can explore for longer. This section culminates with studying how financial support changes agents' behavior, essentially showing that financial support allows agents to expand their exploration horizon while allowing them to still receive comparable reward to their less gritty counterparts. 
Finally, we investigate similar questions in the Bayesian setting, and we find that grit when framed as uncertainty tolerance has a similar effect on an agent's policy, namely that more grit causes an agent to strive for longer.

The rest of the paper is structured as follows. First, we survey related work from several fields, including philosophy, sociology, and computer science. Then, we formally introduce the multi-armed bandits instance we study. In Section~\ref{sec:cr}, we study the competitive ratio-based notion of rationality, following which in Section~\ref{sec:trust-fund} we introduce models for financial support and study the effects of the financial safety net on behavior and reward. Finally, in Section~\ref{sec:bayesian}, we investigate the Bayesian notion of rationality, studying uncertainty tolerance in this setting.

\subsection{Related Work}

In this section, we survey related work at a high level. See Appendix~\ref{appendix:related-work} for more details.

First, we discuss work from the social sciences that observes grit in people and analyzes its role and impact in society.
One of the most popular studies of grit is psychologist Angela Duckworth's book \cite{duckworth2016grit}.
Relatedly, psychologist Martin Seligman has a large body of work encouraging positive attitudes as a path to success and good outcomes \cite{seligman_authentic_2002, seligman_learned_2006}.
More recently, sociologist Tom Wooten studied the impact of the ``no excuses'' educational approach, which draws on the above ideas, in his dissertation \cite{wooten_precarious_2022}, particularly exploring the mechanisms by which these educational systems perpetuate poverty. Abstracting findings from several such observational studies, Morton and Paul develop a philosophical theory of grit \cite{morton2019grit}. We build heavily on the abstractions derived in this work. \par

Next, we survey computer science literature that we draw on in order to quantitatively study the question of decision-making with grit. Our formalizations of rationality are inspired by objectives that are commonly optimized in the computer science literature, including the competitive ratio studied in online learning \cite{borodin_online_2005} and the Bayesian approach to uncertainty quantification  \cite{bernardosmith_bayesian_2000}. The framework we use to represent the decision problem is an instance of the multi-armed bandits problem, a well-studied framework for making decisions when the payoff is unknown \cite{slivkins_introduction_2022}. In particular, we suppose the bandit arms have structured reward, and the structure of interest is improving, first formalized by \cite{heidari} and studied by \cite{patil_mitigating_2023, blum_nearly-tight_2024}.

\subsubsection{Key Features of Grit According to Morton and Paul}

Since we base our development of a theory of grit on the account of \cite{morton2019grit}, it will be helpful to summarize the key points from their work. In their work, Morton and Paul describe gritty behavior as displaying a form of ``epistemic resilience.'' An important part of this is how an agent redefines their goal in the presence of encouraging or discouraging evidence. They also reason that grit is {\em rational}, and therefore agents displaying grit cannot {\em ignore} evidence but rather should be sensitive to failure. Citing works from psychologists Angela Duckworth and Martin Seligman, Morton and Paul further highlight the importance of hope and optimism. They culminate by providing a description of an ``Evidential Threshold'' that captures the decision-making of a gritty agent. The Evidential Threshold as conceptualized by them asks how compelling evidence must be to change the actions of an agent; for a gritty agent, the Evidential Threshold is higher than that of an ``impartial observer.'' Since this threshold hinges on how compelling the agent finds the evidence, Morton and Paul argue that Permissivism applies, and different agents can witness the same evidence but come to different conclusions about what implications that evidence should have on their actions. 
Throughout our paper, we will connect back to these various facets described by Morton and Paul, including by providing a quantitative analog of the Evidence Threshold.

\section{Formal Setting: Improving Multi-Armed Bandits}

We propose studying grit in the improving multi-armed bandits framework \cite{heidari}. 
In this setting, there are $k$ actions (modeled as ``arms''), and at each time step, the agent chooses one of the arms to ``pull,'' i.e., play. 
Each arm has associated with it a 
reward function that increases as a function of the amount of time for which it has been played. This allows us to model situations where engaging with an option increases its payout, for instance when learning a skill or developing a new technology. For most of the paper, we consider a continuous-time version of this framework. Formally:

\begin{definition}
    An instance of the improving multi-armed bandits problem comprises $k$ bandit arms, each of which is associated with a reward function that increases as a function of the (possibly fractional) amount of time for which it has been played. 
\end{definition}

In our models for grit, it is natural to think of options having different payoffs, some of which are static over time and some of which take a while to start paying off but then pay off well once they do. Thus, we propose a two-armed bandit instance in which to study gritty behavior. All of our models will include a stable arm, $f_1(t) = 1 \, \forall \, t$ that represents an option that starts paying off immediately and consistently rewards the agent the same amount (i.e., little scope for growth). The other arm is one in which there is no reward at first (or in certain cases where specified, there is actually a {\em cost} to striving), after which the arm starts providing non-negative reward.
We refer to this arm as the ``striving'' arm, and it provides a reward of 0 units for the first $\theta$ time steps that it is played ($\theta$ being unknown to the agent), following which it increases linearly at a slope of $\alpha\,$ (agents will have beliefs about $\alpha$).
Formally, the two bandit arms are:

\begin{equation} \label{eqn:main-instance}
f_1(t) = 1 \, \forall \, t \qquad f_2(t) = \begin{cases}
    0 & t < \theta \\
    \alpha(t-\theta) & t \ge \theta 
\end{cases}\,,
\end{equation}
where $t$ represents the amount of time for which that arm has been played.

In most of the paper, we use this model, though in some sections we set $\alpha = 1.$ Later on, in Section~\ref{subsec:discomfort}, we also define a notion of ``comfort'' which places a restriction on how often $f_2$ can be played, requiring that $f_1$ be played frequently enough to build up a buffer.

This choice of model is natural: the improving multi-armed bandits problem is well-studied, and there is a clear understanding of what we could hope to achieve in the general case \cite{patil_mitigating_2023, blum_nearly-tight_2024}. The structure in the reward function allows us to capture the fact that {\em investing} time into an option may change its payoff. Finally, the instance described is abstract and flexible, allowing us to model a wide range of real-world settings. On the other hand, a limitation is that this instance only allows for studying an agent's decision between two options. The stable option is arguably over-simplified, since stable options can also lead to growth in the real world. Overall, however, we believe this is a good starting point for formally modelling the decision problem of interest. Further, within the improving two-armed bandit setting, this is the simplest model in which we see non-trivial behavior: if the second arm's payoff were flat instead of linear, the trivial strategy of playing $f_1$ all along would suffice for optimizing the competitive ratio. This is discussed in detail in Appendix~\ref{subsubsec:whynotflat}.

\section{Rationality in Terms of Competitive Ratio} \label{sec:cr}

\subsection{Rationality in This Model}

The first formal model for rationality we study is one in which an agent minimizes regret in hindsight by optimizing the competitive ratio. Competitive ratio measures the relationship between the achieved reward and the optimal reward. This is a commonly-studied notion in theoretical computer science, and particularly in online algorithms \cite{borodin_online_2005}. Achieving a competitive ratio of 1 ensures that the agent did as well as they could hope. Formally, we define the competitive ratio as follows:

\begin{definition}
    An algorithm achieves {\em competitive ratio} $g$ if the ratio of its reward $ALG$ to the optimal achievable reward $OPT$ is at least $g\,,$ i.e., if $ALG/OPT \ge g\,.$
\end{definition}

In certain situations where we are interested in the accumulated reward up to a certain time point $t$ of the algorithm, we will refer to it as $ALG_t\,.$

We will consider agents who have beliefs about $\alpha\,,$ the potential for payoff, and who optimize their competitive ratio over an unknown $\theta\,,$ the time after which the payoff begins. Suppose an agent has a deterministic algorithm to choose how to play arms $f_1$ and $f_2$ as defined above. Each time the agent plays $f_1$ and then switches back to arm $f_2\,,$ they could instead have played $f_2$ followed by $f_1$ while gaining the same reward but possibly witnessing the increase sooner. Thus, there is no benefit to interweaving steps of the arms, and any deterministic strategy for playing this instance can be boiled down to the point at which it switches from $f_2$ to $f_1\,.$ We formalize this in the following lemma.

\begin{lemma} \label{lemma:switchpt}
    Any strategy that interweaves plays of $f_1$ and $f_2$ can be converted into a strategy that plays only $f_2$ followed by only $f_1$ that achieves at least as much reward.
\end{lemma}
\begin{proof}
    Suppose the interweaving strategy plays a total of $t_1$ steps on $f_1$ and $s$ steps on $f_2\,.$ Consider two cases: in the first case, $s < \theta\,.$ Then, the total reward of the policy is $t_1\,,$ and playing $s$ steps of 0-reward-accruing $f_2$ followed by $t_1$ steps of $f_1$ achieves this reward, the same as any interweaving version. In the second case, $s \ge \theta\,.$ Now, by the previous argument, the policy that plays $\theta$ steps of $f_2$ followed by $t_1$ steps of $f_1$ still receives $t_1$ reward. However, after $\theta$ steps on $f_2\,,$ the agent witnesses the increase and therefore has no incentive to switch to the stable arm. The reward of playing the remaining $s-\theta + t_1$ steps on the striving arm is $\frac 12 (s-\theta + t_1)^2\,,$ which is greater than $t_1.$ Thus, we have shown that for each interweaved strategy, there is a non-interweaved one that accrues at least as much reward.
\end{proof}

As a result of this lemma, it suffices to study strategies that play $f_2$ for a while and then permanently switch to $f_1.$ Thus, we will extensively study this switch point, and it will be an interesting quantity to study as a proxy for gritty and non-gritty strategies.  

Morton and Paul discuss an ``Evidential Threshold,'' writing ``In a given context, how much evidence is required -- that is, how compelling must the evidence be --before the thinker comes to a conclusion about what to believe or revises her current beliefs?'' \cite{morton2019grit}. In our proposed framework of analysis, the switch point reflects this evidential threshold. The agent's strategy can be summarized as ``if I don't see evidence that striving is going to pay off until time $s$, I will give up.'' The threshold $s$ is different for different agents, and so the policy the agent follows exactly corresponds to their evidential threshold as described by Morton and Paul.

For a fixed switch point $s\,,$ there are two ``worst'' cases -- the first is when the arms are such that playing the stable arm the whole time would provide the optimal reward, and so the longer the agent spends exploring before switching, the worse the competitive ratio gets. On the other hand, if the striving arm pays off right after the agent switches, then staying on the arm just a little longer would have paid off, so this competitive ratio is increasing with $s\,.$ In order for our strategy to minimize overall regret, we pick $s$ such that the ratio is the same regardless of which extreme case we are in, i.e., we solve for $s$ when the two extreme cases are equal.

\newcommand{\alphat}{\tilde{\alpha}}

\subsection{Modelling Grit: Optimism} \label{subsec:optimism}

Now, let us discuss the relationship between the grittiness of an agent and their approach to the multi-armed bandit problem above. A gritty agent is optimistic about the potential for reward of the risky action, or they would not persevere in taking it.
Accordingly, for this setting, we consider an agent with higher guess for the slope of the increasing portion of arm $f_2$ to be more gritty. Based on this, we can investigate consequences (in terms of both strategy and reward) of demonstrating grit and offer mechanistic insight as to why these consequences exist. The guess for the slope
affects how long the agent is willing to play the striving arm (details in Appendix~\ref{appendix:proof-alphat-switch}). Formally:

\begin{definition}
    In the ``grit-as-optimism'' setting, an agent is $\alphat$-gritty if they guess that the slope of the increasing portion of the striving arm is $\alphat\,.$
\end{definition}


\begin{lemma} \label{lemma:alphat-switch}
    Suppose an agent guesses a value for $\alpha$ that we call $\alphat\,,$ i.e., is $\alphat-$gritty. Assume their goal is to maximize the competitive ratio. Then, they play $f_2$ for $T- \sqrt{\frac{2T}{\alphat}}$ steps, following which they switch to $f_1$ permanently.
\end{lemma}

\subsubsection{Results}

Now, let us consider A, an $\alphat_A$-gritty agent. A is not particularly gritty, so $\alphat_A$ is small. On the other hand, B is somehow privy to perfect information, so $\alphat_B = \alpha\,.$ Finally, consider C, an $\alphat_C$-gritty agent. C is very gritty, and so $\alphat_A < \alphat_B = \alpha < \alphat_C\,.$ (We are simply instantiating the agents in this way to study ``high'' and ``low'' grit as they compare to perfect information.) 
In order to understand the impact the grit-induced strategy has on the reward the agent accrues, we are interested in understanding (1) what level of grit witnesses the striving arm paying off; (2) what level of grit results in good stable reward.

\paragraph*{Observation 1: Duration of Attempt} 
Applying Lemma~\ref{lemma:alphat-switch}, we have that agent A switches at time $s_A = T- \sqrt{\frac{2T}{\alphat_A}}\,,$ agent B at $s_B = T- \sqrt{\frac{2T}{\alpha}}\,,$ and agent C at $s_C = T- \sqrt{\frac{2T}{\alphat_C}}\,.$ Note that $s_C > s_B > s_A\,.$ Our first conclusion, therefore, is that the duration for which an agent explores is longer for a grittier person.

\paragraph*{Observation 2: When Does Increased Grit Benefit the Agent?} Let us now study under what conditions each agent comes out on top. The reward achieved by any agent depends on the relationship between $\theta\,,$ the threshold beyond which the striving arm starts increasing, and $s\,,$ the agent's switch point as stated in the below proposition.

\newcommand{\optreward}{\frac{\alpha}{2} (T-\theta)^2}
\newcommand{\smreward}[1]{\sqrt{\frac{2T}{\alphat_{#1}}}}

\begin{proposition}
    If $\theta \le s\,,$ then an agent switching at $s$ receives $\optreward$ reward, but if $\theta > s\,,$ then the agent receives $T-s$ reward.
\end{proposition}

\begin{proof}
    If $\theta \le s\,,$ then the arm starts paying off while the agent is still playing it. This means that the agent accrues reward starting at time $\theta$ up until time $T$ as the function increases linearly. Hence, the reward is $\frac \alpha2 (T-\theta)^2\,.$ On the other hand, if $\theta > s\,,$ then the agent gains no reward from the striving arm. They gain reward from the stable arm from time $t=s$ to time $t=T\,,$ which is $T-s$ units of reward.
\end{proof}

Let us consider how different levels of grit affect reward\footnote{Figure~\ref{fig:reward-visual} in Appendix~\ref{appendix:reward-theta-vis} may help visualize the relative rewards.}:

\begin{enumerate}
    \item \textbf{Case 1:} $\pmb{\theta < s_A\,.}$ Everyone's reward is the same in this case, since all agents receive $\frac \alpha 2 (T-\theta)^2$ reward.
    \item \textbf{Case 2:} $\pmb{s_A < \theta < s_B\,.}$ In this case, A has given up and switched to the stable arm. As a result, they receive $\smreward{A}$ reward. However, agents B and C stay on the striving arm long enough to witness $\theta\,,$ and so they receive $\optreward$ reward. We see that this is a situation where a lack of grit fares worse than being rather gritty.
    \item \textbf{Case 3:} $\pmb{s_B < \theta < s_C\,.}$ In this case, A and B have both given up and switched to the stable arm, but C valiantly perseveres. Here, A receives $\smreward{A}$ reward, B receives $\smreward{B}$ reward, and C receives $\optreward\,.$ C outshines even B, who had perfect information about the rate of reward increase. In this case, curiously, the least gritty person actually fares better than someone with perfect knowledge of the payoff. We can understand this as follows: there are two reasons why the received reward would be small -- one is location of $\theta$ and one is size of $\alpha$. In this case, someone that is pessimistic about the value of the reward magnitude due to pessimism about $\alpha$ ends up reaping the side benefit when the reward is indeed small, but it is because $\theta$ is large i.e., they are right about the reward being small but for the wrong reasons\footnote{This is an interesting consideration for further modeling -- in this particular setting, the model does not disentangle between these reasons for the reward to be small. More broadly, a competitive ratio-based model that studies reward in general rather than particular kinds of reward may not be able to disentangle this at all.}. 
    \item \textbf{Case 4:} $\pmb{s_C < \theta\,.}$ In this case, no one receives the reward from the striving arm. However, since C has stuck around for so long, they actually also receive less reward overall from the stable arm. This indicates that there is a failure mode when an agent is {\em too} optimistic. Further, the resulting strategy of sticking it out for a long time can fail when $\theta$ is quite large. This reflects what \cite{wooten_precarious_2022} calls the ``effort paradox,'' where students encouraged to be gritty often end up burnt out.
\end{enumerate}

Observe that when the first agent switches, all agents have the same evidence about the striving arm. Likewise, when each agent switches, the remaining agents all have the same evidence, though different agents act differently, albeit all rationally, in response to this information. While it may at first seem disconcerting that people with access to the same evidence have {\em different} yet completely rational responses to it, this is is consistent with the philosophical thesis of Permissivism, which argues that ``some bodies of evidence permit more than one rational doxastic attitude toward a particular proposition'' \cite{lasonen-aarnio_permissivism_2023}. Thus, it is reasonable for different agents to have different beliefs about the underlying state of the world upon receiving a set of evidence (i.e., different agents have different beliefs about how long it is worth staying on an arm after receiving reward 0 for the first $s_A$ time). In fact, Morton argues that Permissivism exactly allows for grit to be conceived of as a factor in shifting people's behavior from the norm even when presented shared evidence. Indeed, in our model, we explicitly model this aspect of how beliefs + evidence $\rightarrow$ actions by having each agent hold different beliefs about the payoff slope $\alpha.$

This perspective also provides us a natural way in which to consider the optimal level of grit for this setting. In particular, if we can remove the uncertainty in the guess for $\alpha\,,$ then the only remaining uncertainty has to do with when the function will start improving. This suggests that agents who have guesses for $\alpha$ closer to the true $\alpha$ will fare better. We could interpret this as having expert advice or access to ``better'' information. 

\subsection{Modelling Grit: Discomfort Tolerance} \label{subsec:discomfort}

In this setting, we study another relevant aspect of grit, namely discomfort tolerance. To do so, let us first introduce a cost to playing the striving arm.

\paragraph*{Cost to strive}

Let us consider a setting in which there is a cost to strive. The reward profile / bandit arms $f_1$ and $f_2$ look almost as before, with $\alpha = 1\,$:
$f_1(t) = 1 \, \forall \, t$ and $f_2(t) = \begin{cases}
    -1 & t < \theta \\
    t-\theta & t \ge \theta
\end{cases}\,.
$

The negative reward models the cost of striving, for instance financial debt or effort expenditure that can be offset by rewards from the stable arm to ensure the agent is not ``in the negative,'' as
the agent is not allowed have negative reward at any point in time. This corresponds to not being able to take out a loan. Thus, an agent who does not have any ``savings'' coming in must play $f_1$ before they ever play $f_2\,.$ We suppose that smallest piece of time an agent can split between the arms is 1 unit. Later, in Section~\ref{sec:trust-fund}, we consider what happens when the cost of striving is subsidized by a ``trust fund.''

\paragraph*{Comfort}
Now, further, let us suppose an agent always wants to have average reward at least $\gamma\,,$ that is, at time $t\,,$ the agent wants their accrued reward to be at least $\gamma \cdot t\,.$ They still aim to optimize the competitive ratio as before, except now they do so subject to the constraint that $\forall \, t\,, $ the reward accrued up to that point is at least $\gamma \, \cdot \, t\,.$

\begin{definition}
    We say an agent wants to be $\gamma$-comfortable if at any time $t > 0\,,$ they need their net reward to be at least $\gamma \cdot t\,.$
\end{definition}

In general, the problem an agent playing this game solves is:
$$
\max \qquad  \frac{ALG}{OPT} \quad \text{ s.t. } \quad \frac{ALG_t}{t} \ge \gamma\,.
$$

Since the agent can play an arm for fractional amounts of time, an agent who requires $\gamma$ comfort will play $\alpha_\gamma$ time on the stable arm and then $1-\alpha_\gamma$ time on the striving arm, alternating between the two as soon as possible, for $\alpha_\gamma \coloneqq (\gamma + 1)/2$\,, the time duration that guarantees the desired average reward:
\begin{align*}
    ALG_t = \begin{cases}
        t & t \le \alpha_\gamma \\
        \alpha_\gamma-(t-\alpha_\gamma) = 2\alpha_\gamma - t & \alpha_\gamma < t \le 1
    \end{cases} \qquad &\Rightarrow \qquad 
    \frac{ALG_t}{t} = \begin{cases}
        1 & t \le \alpha_\gamma \\
        \frac{2\alpha_\gamma}{t} - 1 & \alpha_\gamma < t \le 1
    \end{cases} \\
   \text{for } \alpha_\gamma < t \le 1\,, \qquad \gamma = 2 \alpha_\gamma - 1 &\le \frac{2\alpha_\gamma}{t} - 1 \le 1\,.
\end{align*}

This is the most rational thing for them to do, since any additional steps on $f_1$ simply serve to offset future costs of $f_2$ but might prevent the agent from seeing the increase phase as quickly. We formalize this below and provide the proof in Appendix~\ref{appendix:proof-min-acc} (which proceeds via similar casework to the proof of Lemma~\ref{lemma:switchpt}).

\begin{definition}
    We call a strategy {\em minimally accumulating} for an agent who wants to be $\gamma$-comfortable if the average net reward at time $t\,,$ $ALG_t/t$ is exactly 1 until reaching time $\alpha_\gamma \coloneqq (\gamma + 1)/2$ when playing the stable arm and strictly decreasing until reaching {\em value} $\gamma$ when playing the striving arm. In other words, the agent plays the stable arm for $\alpha_\gamma$ time followed by the striving arm for $1-\alpha_\gamma$ time and repeats. 
\end{definition}
\begin{lemma} \label{lemma:min-acc-works}
    For each strategy that ``stockpiles'' reward along the way, there exists a minimally accumulating strategy that nets total reward at least as much as the stockpiling strategy. 
\end{lemma}

We present the result for the competitive ratio and reward for a $\gamma$-comfortable agent (proof in Appendix~\ref{appendix:proof-gamma-comfort}).

\begin{lemma} \label{lemma:gamma-comfort}
Suppose 
an agent who wants to be $\gamma$-comfortable plays $f_1$ for $\alpha_\gamma$ time followed by $f_2$ for $1-\alpha_\gamma$ time before reverting back to $f_1\,$ and continuing the process. Then, the agent wanting to maximize their competitive ratio subject to the constraint of the average reward always being at least $\gamma$ will switch after absolute time $T - \frac{\gamma}{2} - \frac12 \sqrt{\gamma^2 + 4T(2-\gamma)}$, achieving competitive ratio $\gamma + \frac{\gamma(1-\gamma)}{2T} + \frac{(1-\gamma) \sqrt{\gamma^2 - 4T (2 - \gamma)}}{2T}$. This corresponds to $\frac{1-\gamma}{2} \cdot \left(T - \frac{\gamma}{2} - \frac12 \sqrt{\gamma^2 + 4T(2-\gamma)} \right)$ time on the striving arm.
\end{lemma}

\begin{remark}
    Let us consider a concrete numerical example: suppose $T = 150\,, \gamma = 0.5.$ Then, the agent will switch after about time 135. However, of that time, only about 34 would have been spent exploring, with the remaining 101 time spent on the stable arm. The competitive ratio achieved is 0.55.
\end{remark}

\begin{remark}
    
For insight, let us next consider the behavior for extreme values of $\gamma\,:$ if $\gamma = 0\,,$ the first two terms go away, and the competitive ratio is $\frac{\sqrt{2T}}{T} = \frac{2}{\sqrt{T}}\,.$ This exactly aligns with our computation before. On the other hand, if $\gamma \rightarrow 1\,,$ the competitive ratio actually nears 1! This is because the best possible thing to do for an agent who requires complete comfort is to always play the stable arm. Indeed, to better understand this outcome, let us also investigate the total amount of time spent exploring as a function of $\gamma\,.$ Taking the derivative of the expression for exploration time with respect to $\gamma$, we can see that it is negative for all $\gamma \in [0, 1]\,.$ 
This tells us that as an agent requires more ``comfort,'' they spend less time exploring on the striving arm, and so while their competitive ratio improves, their chance of witnessing the growth in the striving arm and benefiting from it is low.
\end{remark}
\section{Financial Support} \label{sec:trust-fund}

A natural question following this discussion is how to incentivize a gritty agent to explore for longer. We have already seen that simply encouraging more grit could do more harm than good. We could encourage the agent to give up more comfort, but often a baseline level of comfort cannot be foregone -- for instance, one may have to pay rent, buy food, etc. For this, we study the setting where there is a cost to striving and the agent must have discomfort tolerance $\gamma > 0$. In this section, we consider an implementation of a ``trust fund,'' i.e., financial support we could provide an agent, and then show what conclusions we can draw about the strategy followed by an agent and their eventual reward. Here we consider the simplest possible implementation, which already has interesting outcomes, and we study a more nuanced implementation in Appendix~\ref{appendix:fixed-time-support}.

\subsection{No safety net} Now, since an agent without a safety net must alternate between $f_1$ and $f_2\,,$ the competitive ratio maximizing strategy is computed as follows:

\begin{align}
    \text{competitive ratio if never increase } &= \frac{(1 - 1) \cdot \frac s2 + T - s}{T} \\
    \text{competitive ratio if increase right after switch } &= \frac{(1 - 1) \cdot \frac s2 + T - s}{(1 - 1) \cdot \frac s2 + \frac{1}{2} (T-s)^2}\,. \\
    \text{equalizing, } 
    s &= T  - \sqrt{2T}\,.
\end{align}

The duration of time after which the switch happens is $T-\sqrt{2T}\,,$ but the proportion of that time actually spend exploring the striving arm is half that, namely $\frac{T-\sqrt{2T}}{2}\,.$ Thus, if $\theta > \frac{T-\sqrt{2T}}{2}\,,$ this agent will not be able to reap the benefits of striving.

\subsection{Free reimbursement} In this model, an agent with a support network gets reimbursed for free each time they taking a striving step. This is analogous to having a benefactor who financially supports the agent as much as needed to prevent their net reward from being negative. In this case, the effective arms for an agent with such unconditional support are now:

$$
\hat{f}_1(t) = 1 \, \forall \, t \qquad \hat{f}_2(t) = \begin{cases}
    0 & t < \theta \\
    t-\theta & t \ge \theta
\end{cases}\,.
$$

Thus, the analysis is as before (Lemma~\ref{lemma:alphat-switch} with $\alpha, \alphat=1$), and the agent will spend $T - \sqrt{2T}$ time on the striving arm before switching to the stable arm. Remarkably, the duration of time after which the agent without a safety net and the agent with unconditional support ``give up'' is the same! However, due to the safety net, the agent with it can explore the striving arm for twice as long. In particular, if $\theta \in [\frac{T-\sqrt{2T}}{2}, T -\sqrt{2T}]\,,$ then the latter agent gets $\frac 12 (T - \theta)^2 \ge \frac 12 \cdot 2T = T$ reward, while one without the safety net only gets $\sqrt{2T}$.

In Appendix~\ref{appendix:fixed-time-support}, we extend this to a model where the benefactor only promises support for a fixed amount of time. We show that in that case also, a similar qualitative result holds -- agents with and without support ``give up'' on striving at the same absolute time but the agent with the financial support gets to explore for a multiplicative factor longer. Note also that by mapping this onto the discomfort tolerance perspective, we can see that in a world where the agent must maintain a positive discomfort tolerance but the agent lacks financial support, they must spend less time exploring, whereas if they have financial support, they can trivially explore longer.

\subsection{Which wins? Trust fund or grit?}

In this section, we combine the pieces from the previous sections to understand the interplay between grit and financial support. We primarily present conclusions in this section; the calculations behind these conclusions are presented in detail in Appendix~\ref{appendix:combine}. Here, the rate of increase is unknown, and there is also a cost to striving.
Let us make two comparisons. First, let us investigate what happens when someone without a trust fund becomes more gritty. Then, we will study the effect of a trust fund on two people, one of whom is grittier than the other.

\textbf{\underline{No trust fund, increased grit.}} For this, let us compare the first two rows of Table~\ref{tab:my_label1}. Immediately, we see that increased grit, as before, leads to increased exploration time. In the last column, we present the reward that the agent receives if they don't witness the start of the payoff before switching, which we call ``stable reward'' for short. There, we can also see that the stable reward is {\em lower} when the agent is grittier. 

\textbf{\underline{Introduce trust fund, same grit.}} Now, let us compare the first and third rows of Table~\ref{tab:my_label1}. In this case, we see that at the same grit level, the presence of a safety net allows for a much longer exploration horizon. Both in the presence of and in the absence of the safety net, the stable reward is the same. This shows that the presence of a safety net allows for essentially ``free'' exploration.

\subsubsection{Discussion} We can view these results from two perspectives. One perspective is descriptive: we observe that providing a safety net increases exploration time essentially ``for free.'' This happens since the agent does not have to split their time to ensure they are not in the negative. The other perspective is prescriptive: in settings where an external entity aims to encourage exploratory behavior, encouraging increased grittiness could lead to worse outcomes if the risk doesn't pay off. While we don't study this in this work, this could lead to future agents being discouraged from taking the grittier course of action. On the other hand, providing support to an agent who is already gritty encourages exploration without the prospect of bad outcomes in case the taken risk doesn't pay off. Thus, telling a gritty agent to be grittier is worse than giving them financial support to extend their exploration horizon. This corroborates and provides a simple mechanistic explanation for what \cite{morton2019grit, wooten_precarious_2022} observe.

\begin{table}[h]
    \centering
    \begin{tabular}{c|c|c|c}
        grit level & safety net & exploration time & stable reward \\
        \hline 
        $\alphat_1$ & no safety net & $\frac{T}{2} - \sqrt{\frac{T}{2\alphat_1}}$ & $\sqrt{\frac{2T}{\alphat_1}}$\\
        $\alphat_2 > \alphat_1$ & no safety net & $\frac{T}{2} - \sqrt{\frac{T}{2\alphat_2}}$ & $\sqrt{\frac{2T}{\alphat_2}}$ \\
        $\alphat_1$ & safety net, free reimbursement & $T - \sqrt{\frac{2T}{\alphat_1}}$ & $\sqrt{\frac{2T}{\alphat_1}}$\\
    \end{tabular}
    \caption{Exploration time and stable reward (reward resulting from switching to $f_1$ following unsuccessful exploration on $f_2$) for various grit levels and safety nets.}
    \label{tab:my_label1}
\end{table}

\section{Rationality as Being Bayesian} \label{sec:bayesian}

\subsection{Rationality in This Model}

Until now, we have focused on a model of rationality that aims to minimize regret in hindsight, which we quantify through the competitive ratio. In this section, we instead consider a notion of rationality that is forward-looking, as it aims to directly quantify uncertainty about the future.  In particular, we suppose agents have a prior distribution $P$ that is supported on $\{1, 2, \dots, T-1, T\} \cup \{ N \},$ where $N$ represents ``never,'' and $P(x) \coloneqq \mathbb{P}\left[\text{increase occurs at time } x\right]\,.$

The agent updates their posterior as follows. Suppose at time $t-1$, the posterior is $P^{(t-1)}.$ If they play the arm for one time step and the increase has not yet occurred, then they update:

$$
P^{(t)}(x) \leftarrow \frac{P^{(t-1)}(x)}{1 - P^{(t-1)}(t)} \qquad \forall x > t\,.
$$

If the mass on the numerical elements of the support is 0, then the update places probability 1 onto the support element $N\,,$ meaning the agent's posterior suggests the arm will {\em never} payoff.

\subsection{Setting}
As before, there are two arms:
$$
f_1(t) = 1\, \forall \, t \qquad f_2(t) = \begin{cases} 0 & t<\theta \\ t-\theta & t\ge \theta
\end{cases}
\,.
$$

Unlike before, we assume the arm is played in discrete time steps of size 1, i.e, the agent pulling the arm increases the time by 1. This is so that the prior can be defined over a discrete domain.

Based on this and the posterior update described above, we have that the recursive formula for the expected reward is:
$$
Q(t) = \frac{1}{2}(T-t)^2 \cdot p_t + V(t+1) \cdot (1-p_t)\,,
$$
and reward of the policy of being on the striving arm given no increase so far and haven't switched so far is:
$$
V(t) = \max \{ T-t, Q(t)  \}
$$

where  $p_t = P(\text{ increase starts at time } t | \text{ increase has not started up until } t-1)\,.$ The boundary condition is $Q(T) = 0\,.$ The agent switches when the reward from switching exceeds the expected reward from staying.

\begin{lemma}
    Suppose $Q(t), V(t)$ are defined as above. Then, $V(t)$ computes the total reward accrued from time $t$ up to time $T$ by a policy that maximizes expected reward over its posterior.
\end{lemma}
\begin{proof}
We show this using induction. From the boundary condition, we have that $V(T) = \max\{0, 0 \} = 0.$ For the base case, we consider $t = T-1\,.$ There are two possible actions the policy could take. Suppose the policy plays the striving arm: the first possible outcome is that it pays off with probability $p_{T-1}$ as defined above, and the second possible outcome is that it doesn't, and the policy accrues reward $V(T)$. We have that $Q(T-1) = \frac 12 p_{T-1} + 0 (1-p_{T-1}) = \frac{p_{T-1}}{2}.$ If the policy does not play the striving arm and instead switches to the stable arm for good, it is guaranteed $T-(T-1) = 1$ reward. Thus, an expected reward maximizing policy will accrue $V(T-1) = \max \{1, \frac{p_{T-1}}{2}\} = 1\,$ reward at that time.

Next, consider the inductive assumption that for $t = t' < T\,, V(t')$ gives the correct total reward accrued from $t'$ to $T.$ Then, we must show that $V(t'-1)$ indeed is the correct total reward accrued from $t'-1$ to $T\,.$ If we play the stable arm from here onward (and we know from Lemma~\ref{lemma:min-acc-works} that once we switch to the stable arm, we have no reason to switch back to the striving arm), then we would get reward $T-(t'-1) = T - t' + 1.$ On the other hand, if we play the striving arm for one step, then if it pays off immediately (probability $p_{t'-1}$), we get reward $\frac 12 (T-t'+1)^2\,,$ and if not yet (probability $1-p_{t'-1}$), we get reward $V(t')\,.$ Since we know that $V(t')$ is the correct total reward accrued from $t'$ to $T$ by the inductive assumption, $Q(t'-1) = \frac 12 (T-t'+1)^2 \, p_{t'-1} + V(t') \, (1-p_{t'-1})$ is the correct total reward accrued from $t'-1$ to $T$ if the policy plays the striving arm. Finally, we do choose the action that maximizes expected reward, and so $V(t'-1) = \max \{ T-t'+1, Q(t'-1)  \}$ is indeed the correct total reward accrued from $t'$ to $T\,.$
\end{proof}

In this view, we may study the uncertainty tolerance aspect of grit. If two agents have similar priors, i.e., same family of distribution and same mean, then the variance of the prior reflects how tolerant they are to uncertainty about when the striving arm will pay off. We primarily study how this uncertainty tolerance affects the policy an agent follows. We can also then understand how the reward is affected as before.

\subsection{Results for Modelling Grit as Uncertainty Tolerance}

We consider the case where an agent has a Gaussian prior on when the striving arm will pay off. This corresponds to having some understanding of {\em when} the arm might start paying off but not being entirely sure that it will pay off then so allowing for some latitude. Thus, the variance of the Gaussian prior corresponds to the uncertainty tolerance of the agent.

We find that as the variance of the prior increases, the point $s$ at which the agent reverts to the stable arm also increases. Thus, as an agent becomes more gritty, there is a wider range in which if $\theta$ lies, they will witness it. At the same time, if $\theta$ lies above the blue curve in Figure~\ref{fig:gauss-bayes}, the agent would not witness the increase, and they would also collect less stable reward. From the asymptoting shape, we can conclude that beyond a certain point, uncertainty tolerance has limited benefits, as the stable reward decreases but the additional region is not increasing by much.

\begin{figure}
    \centering
    \includegraphics[width=0.5\linewidth]{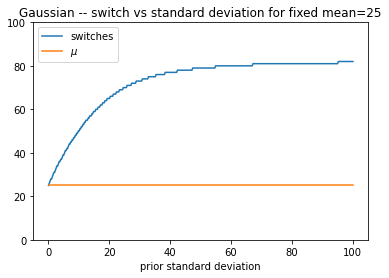}
    \caption{Switch point as a function of standard deviation $\sigma$ for prior of $\mathcal{N}(25, \sigma^2).$}
    \label{fig:gauss-bayes}
\end{figure}

Morton and Paul note that:
\begin{quote}
    ``Other things being equal, the gritty agent’s evidential threshold for updating her expectations of success will tend to be higher than the threshold an impartial observer would use. This is not because the perspective of the impartial observer is epistemically privileged, however; the Permissivist latitude applies to the policies of the agent and the observer alike. Rather, it is because the observer has no need to respond to the evidence in a way that guards against premature despair, and this should be reflected in his evidential policies.'' (p. 195 in \cite{morton2019grit})
\end{quote}

Our results corroborate this -- if we suppose the `impartial observer' switches if the striving arm has not yet paid off when $s$ reaches their expectation for when they payoff begins, then the gritty agent indeed has a higher threshold for updating their expectation of success.

\section{Discussion and Conclusion}

\subsection{Discussion of Modelling Choices}

\textbf{Modelling Perspective} In this work, we provide a theoretical model for the following aspects of grit: we define a two-armed improving bandit instance that represents the decision-making problem agents are faced with. We further define two different objectives a decision-making agent might prioritize. Finally, we model different aspects of grit with careful study of variations of parameters in the bandit instance. While there are many folk notions of “grit”, the aspects we choose to model are based on the abstraction and formalization provided by \cite{morton2019grit}. We do not seek to evaluate whether or not their account is valid but rather to quantitatively formalize their notions and study the implications of their analysis. 
In doing so, we find additional confirmation for empirical work regarding the value of ``grants'' provided to people attempting ambitious choices. For instance, a study of grants given to entrepreneurs in Burkina Faso shows that even if immediate profits aren't improved, providing this financial support increases innovation and improves business practices, suggesting that those who get support get more time to explore and build a good foundation for their business \cite{grimm_short-term_2021}. Similarly, a study in Kenya found that intervening with grants to youth entrepreneurs at a time of crisis helped individuals maintain their businesses and produce more profits \cite{domenella_can_2021}. These empirical studies suggest that if someone is inclined to take an ambitious action, then supporting them helps improve outcomes, exactly what is captured by our model (as summarized in Table~\ref{tab:my_label1}).

\textbf{Strengths} A strength of our model is that the same multi-armed bandits instance can be used to study many different facets of grit, and resulting outcomes are visible just through this two-armed bandits instance. Also, our modelling allows for a modular understanding of the effect of grit on first, behavior and second, outcome or reward associated with that behavior. This, then, allows us to study interventions that encourage similar behavior with less outcome risk. 

\textbf{Weaknesses} Our model is not without its shortcomings. As mentioned before, we are limited to studying an agent's choice between two options in this framework, and often the stable option does not just have the same return for all time. Human rationality is rarely, if ever, executed exactly as competitive ratio maximization or expected reward maximization over a prior. However, these seem like natural simplifications that provide a starting point for quantitative analysis of this trait.

\subsection{Conclusion}
In this work, we introduced a quantitative model for studying the impact of grit and financial support on decision-making between a long-term, ambitious end and immediate-reward, stable end. We provided a improving two-armed bandit model in which to study this and formalized notions of rationality and grit that gave rise to agents' strategies for this instance. Our modeling engages with several prior works in the social science literature, and we hope the approach and framework provided admit future work on quantitative study of grit.
\newpage
\printbibliography

\clearpage

\appendix
\section{Related Work} \label{appendix:related-work}

We survey two categories of work related to ours. First, we discuss work from the social sciences that observes grit in people and analyzes its role and impact in society. We also include work that abstracts and formalizes what differentiates grit from other related traits. Next, we focus on relevant computer science literature related to the formal frameworks we use in our analysis.

One of the most popular studies of grit is psychologist Angela Duckworth's book \cite{duckworth2016grit}. In it, she discusses how across fields and circumstances, the trait of grit distinguishes people who succeed from people who succumb to their circumstances. This work had a significant impact on popular understanding of the impact of hard work on success, suggesting that grit is often a more important factor than talent or proclivity. Relatedly, psychologist Martin Seligman has a large body of work encouraging positive attitudes as a path to success and good outcomes \cite{seligman_authentic_2002, seligman_learned_2006}. School systems like the Knowledge is Power Program (KIPP) Charter Schools have operationalized this somewhat academic research (including by bringing Seligman on as a consultant \cite{Gibbon_seligman_2020}) to build their signature ``no excuses'' philosophy of education, emphasizing personal responsibility and grit as paths to success. 

More recently, sociologist Tom Wooten studied the impact of the ``no excuses'' educational approach in his dissertation \cite{wooten_precarious_2022}, using ethnographic methods to study outcomes for six students who attended schools with this educational philosophy, particularly exploring the mechanisms by which these systems perpetuate poverty. In particular, he argues that upward mobility in the United States is fundamentally precarious, with seemingly small disruptions destabilizing those coming from impoverished backgrounds much more than those who have more financial stability. Wooten's findings are the impetus for us to investigate the impacts of financial support on the behavior of gritty agents in Section~\ref{sec:trust-fund}. He also reflects on the pitfalls of an educational system that overemphasizes grit, writing that, ``In lieu of such large, fundamental changes, my research also points to smaller policy shifts that could make a big difference. One shift is that primary and secondary schools should carefully examine the explicit and implicit lessons they teach their students about
`grit,' altering their curricula to emphasize balance and pacing.'' \par

Abstracting findings from several such observational studies, \cite{morton2019grit} develop a philosophical theory of grit, arguing that it is possible due to a philosophical concept called \textit{Permissivism} for two agents to interpret the same body of evidence in completely different ways, and in fact grit is a consequence of Permissivism. We build heavily on the abstractions derived in this work. \par

Finally, we survey some computer science literature that we draw on in order to quantitatively formalize and study the question of decision-making with and without grit. The perspectives we take on rationality are inspired by objectives that are commonly optimized in the computer science literature, including the competitive ratio studied in online learning \cite{borodin_online_2005} and the Bayesian approach to uncertainty quantification  \cite{bernardosmith_bayesian_2000}. The framework we use to represent the decision problem is an instance of the multi-armed bandits problem, a well-studied framework for making decisions when the payoff is unknown \cite{slivkins_introduction_2022}. More particularly, we suppose the bandit arms have structured reward, and the structure of interest is improving, first formalized by \cite{heidari} and studied in general by \cite{patil_mitigating_2023, blum_nearly-tight_2024}.
\section{Why This Instance and Not Something Else?} \label{subsubsec:whynotflat}

The instance defined in Equation~\ref{eqn:main-instance} reflects settings where grit pays off eventually but even then not immediately, i.e., there is only a gradual increase of reward even after $t = \theta\,.$ This captures certain situations we may be interested in -- for instance, consider a computer science theory PhD program, where the first several years are spent taking classes to bolster one's mathematical skills, following which there is still a ramp up period as the student gets accustomed to doing research. Thus, it is natural to suppose the early years of PhD candidacy yield some research progress, while the later years yield significantly more. On the other hand, there are many natural cases in which once the payoff starts, it reaches its complete potential immediately; a natural example of this is when well-digging or searching for oil. In these cases, a bandit arm formulation more like the following may be more reasonable:
$$
f_1(t) = 1 \, \forall \, t \qquad f_2(t) = \begin{cases}
    0 & t < \theta \\
    m & t \ge \theta
\end{cases}\,.
$$

If we aim to maximize the competitive ratio, we consider two extreme worst cases: first, one in which the agent spent $s$ time on the striving arm but instead should have spent all their time on the stable arm; second, one in which the agent spent $s$ time on the striving arm to no avail but would've witnessed the increase on the striving arm if they'd stuck around infinitesimally longer. In this case, the former is $(T-s)/T\,,$ which is decreasing with $s\,,$ while the latter is $(T-s)/(m(T-s)) = 1/m\,,$ which is non-increasing. This implies that as long as the agent switches to the stable arm before $s = (1-1/m)T\,$, they achieve competitive ratio at least $1/m\,.$ 
Thus, achieving the optimal competitive ratio is trivial and there is no strategy to the game. Owing to this, we study a slightly more complicated setting -- the instance in which $f_2$ increases linearly once it starts paying off.

In general, the non-trivial instances are ones in which the first competitive ratio extreme case is decreasing with $s$ and the second competitive ratio extreme case is increasing with $s\,.$ If we define $F_2 \coloneqq \int f_2\,,$ we find $F_2$ needs to be increasing with $s$ in order to find a sensible switch point that balances the worse cases, i.e.:
$$
\frac{T-s}{T} = \frac{T-s}{F_2(T-s)} \Rightarrow s = T - F_2^{-1}(T) \Rightarrow CR = \frac{F_2^{-1}(T)}{T}\,.
$$

\section{Proofs of Results}

\subsection{Proof of Lemma~\ref{lemma:alphat-switch}} \label{appendix:proof-alphat-switch}
\begin{proof}

Recall our previous argument that all deterministic strategies boil down to playing $f_2$ for $s$ steps and then switching to $f_1$. With that in mind, there are two extremes of what could happen to the competitive ratio, one of which decreases with the longer spent on $f_2$, and the other of which increases with time spent on $f_2\,.$ Thus, we compute them and equalize:
\begin{align}
    \text{competitive ratio if never increase } &= \frac{T-s}{T} \\
    \text{compeititve ratio if increase right after switch } &= \frac{T-s}{\frac{\alphat}{2} (T-s)^2}\,. \\
    \text{equalizing, } \frac{T-s}{T} &= \frac{T-s}{\frac{\alphat}{2} (T-s)^2} \\
    s &= T - \sqrt{\frac{2T}{\alphat}}\,.
\end{align}

\end{proof}

\subsection{Visual Depiction of Reward Analysis} \label{appendix:reward-theta-vis}
See Figure~\ref{fig:reward-visual}.
\begin{figure}[h!]
    \centering
    \includegraphics[width=0.85\linewidth]{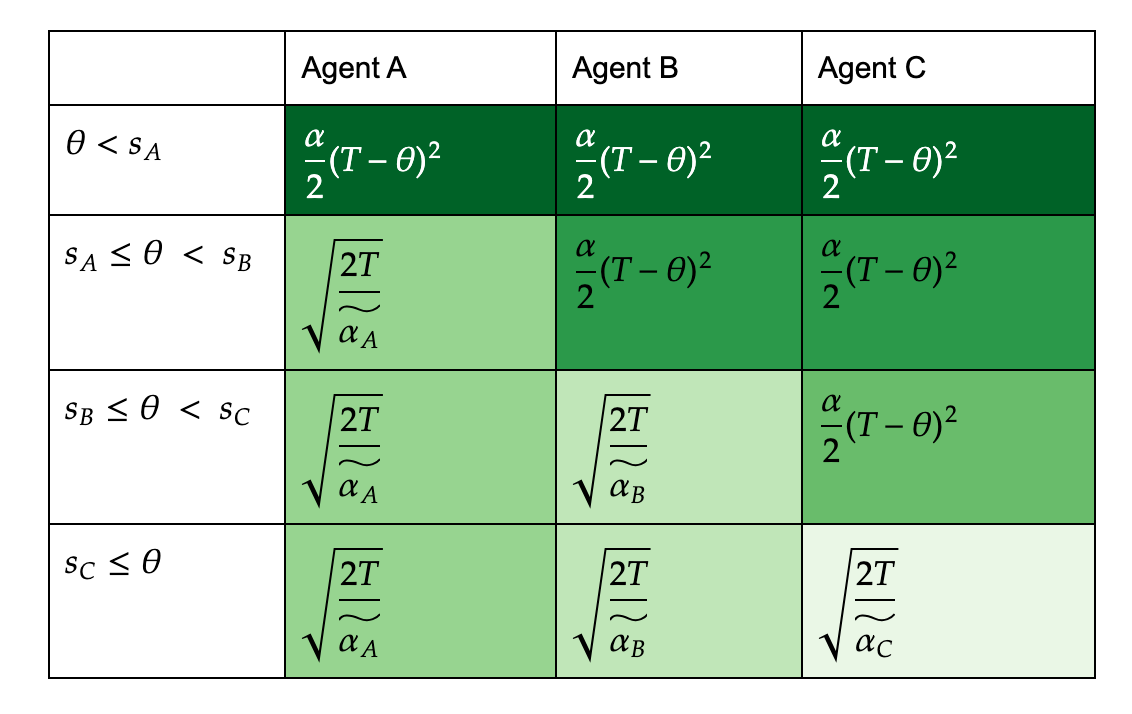}
    \caption{This table shows the reward the agents described in Section~\ref{subsec:optimism} receive if $\theta\,,$ the true threshold beyond which $f_2$ pays off, lies in different regions. The columns are increasing in grit from left to right. We can see that if $\theta$ is larger, more grit can actually result in less reward, since the agent switches back to the stable arm and accrues less reward there.}
    \label{fig:reward-visual}
\end{figure}

\subsection{Proof of Lemma~\ref{lemma:min-acc-works}} \label{appendix:proof-min-acc}

\begin{proof}
    For a strategy, let $\{t_i\}_{i = 0}^n$ be the set of times at which the agent switches from one arm to the other with $t_0 = 0$. Then, $[t_i, t_{i+1}]$ for even $i$ are the intervals played on the stable arm, and the remaining intervals are played on the striving arm.
    We argue that we can convert any strategy that is not-minimally accumulating into a minimally-accumulating strategy that accrues at least as much reward as the original non-minimally-accumulating strategy. 

    Suppose the total amount of time spent on the stable arm $f_1$ is $T_1 \coloneqq \sum_{i \text{ even}} (t_{i+1} - t_i)\,.$  and the total time spent on striving arm $f_2$ is $T_2 \coloneqq \sum_{i \text{ odd}} (t_{i+1} - t_i)\,,$ such that $T_1 + T_2 = T\,.$ Recall $\theta\,,$ the threshold until which one must play $f_2$ before it pays off. If $T_2 < \theta\,,$ then redefining the intervals such that the strategy becomes minimally-accumulating does not change the total accumulated reward, since the reward accrued from an arm is only a function of the time {\em on that particular arm}. Thus, the total reward is $T_1$ regardless of the order in which it is accrued.

    Let us now consider the second case, $T_2 > \theta\,.$ Now, let us break up the time spent on arm $f_1$ into the $\frac{1+\gamma}{1-\gamma} \, \theta$ time that must happen prior to playing the striving arm for $\theta$ time and the remaining $T_1 - \frac{1+\gamma}{1-\gamma} \, \theta$ time. (Due to the comfort restriction, we know $T_1 - \frac{1+\gamma}{1-\gamma} \, \theta \ge 0\,.$) As argued previously, we can rearrange interval endpoints for the first $\frac{1+\gamma}{1-\gamma} \, \theta$ time spent on $f_1$ and $\theta$ time spent on $f_2$ with no changes in overall reward. In particular, we can rearrange it to be minimally accumulating. Now, for the remaining $T_1 - \frac{1+\gamma}{1-\gamma} \, \theta$ time, if we play it on the stable arm (at any time), we accrue $R = T_1 - \frac{1+\gamma}{1-\gamma} \, \theta$ reward. However, having satisfied the comfort requirements, if we instead use it at the end of the game and on arm $f_2\,,$ we instead gain reward $(T_2 - \theta - R/2)R\,.$ If $R/2 + T_2 - \theta > 1\,,$ then this is strictly better. If not, we can still play this time on the stable arm. Thus, we have constructed a minimally-accumulating strategy that accrues at least as much reward as the original strategy.

\end{proof}

\subsection{Proof of Lemma~\ref{lemma:gamma-comfort}} \label{appendix:proof-gamma-comfort}

\begin{proof}
    As before, we compute the competitive ratio in two cases: one in which the stable arm played the whole time would be optimal and one in which the striving arm pays off as soon as the agent switches for good. Since the agent is alternating between $\alpha_\gamma \coloneqq \frac{\gamma+1}{2}$ time on $f_1$ and $1-\alpha_\gamma$ time on $f_2\,,$ we must account for this in the computation:

    $$
    \frac{\alpha_\gamma \cdot s - (1-\alpha_\gamma)s + T - s}{T} \, = \, \frac{\alpha_\gamma \cdot s - (1-\alpha_\gamma)s + T - s}{\frac12 (T-s)^2 + \alpha_\gamma \cdot s - (1-\alpha_\gamma)s}\,.
    $$

Solving, we get:
\begin{align}
    s^2 - 2sT + T^2 + s(2\alpha_\gamma-1) - 2 T &= 0 \\    
    s &=  T - \alpha + \frac12 -\frac12 \sqrt{(2\alpha_\gamma - 1)^2 + 4T(3-2\alpha)} \\
     &= T - \frac{\gamma}{2} - \frac{\sqrt{\gamma^2 + 4T(2 - \gamma)}}{2}\,.
\end{align}

We can plug this back in to get the competitive ratio:
$$
\frac{(2\alpha_\gamma - 1)s + T - s}{T} = \gamma + \frac{\gamma(1-\gamma)}{2T} + \frac{(1-\gamma) \sqrt{\gamma^2+ 4T (2-\gamma)}}{2T} \,.
$$

Finally, the amount of time spent exploring on the striving arm is $1-\alpha_\gamma$ fraction of the total time pre-switch, which is:
$$
\frac{1-\gamma}{2} \cdot \left(T - \frac{\gamma}{2} - \frac{\sqrt{\gamma^2 + 4T(2 - \gamma)}}{2} \right)\,.
$$

\end{proof}

\section{Fixed Time Financial Support}
\label{appendix:fixed-time-support}
Now, suppose the unconditionality of support is a little weaker. Instead of ``as much as needed,'' the benefactor only promises $R$ units of support. This could correspond to a parent committing to financially support a child until 18, or 21, or 25, etc. Now, the calculus is a little different. First we consider the case where $R$ is large, say like $T\,.$
\begin{align}
    \text{competitive ratio if never increase } &= \frac{R - s + T - s}{R + T} \\
    \text{competitive ratio if increase right after switch } &= \frac{R - s + T-s}{R - s + \frac{1}{2} (T-s)^2}\,. \\
    \text{equalizing, } \frac{R - s + T - s}{R + T} &= \frac{R - s + T-s}{R - s + \frac{1}{2} (T-s)^2} \\
    s = T + 1 - \sqrt{4T + 1}\,.
\end{align}

Interestingly, in this model the agent switches to the stable arm a little sooner than in the previous case -- which makes sense as the agent knows not to expect unbounded financial support. However, the agent can still explore for that full time owing to the financial support. Thus, when compared to an agent without financial support, this agent can explore for a longer time by a factor of $\frac{T+1 - \sqrt{4T + 1}}{\frac{T - \sqrt{2T}}{2}} \ge \frac{2(T - \sqrt{5T})}{T - \sqrt{2T}}\,,$ which is at least 1.5 for large enough $T$ ($T \ge 23$).

These results allow us to conclude that even when all agents have nearly identical ``outward'' behavior, i.e., the point at which they finally give up on striving and switch to the stable arm, the presence of a trust fund allow agents to explore for longer, i.e., spend more time on the striving arm. Our result provides a quantitative mechanism by which we can explain why agents who seemingly spend similar amounts of time on ambitious goals can see significant differences in outcome when they come from different kinds of resource backgrounds.

\section{Calculations For Interplay of Grit and Trust Fund} \label{appendix:combine}

Now, we look at what happens when the rate of increase of the arm is unknown so some amount of grit plays into the decision of when to switch, but there is also a cost to striving and net reward is never allowed to be negative. Formally, at each time step, the agent chooses between:

$$
f_1(t) = 1 \, \forall \, t \qquad f_2(t) = \begin{cases}
    -1 & t < \theta \\
    \alpha (t-\theta) & t \ge \theta
\end{cases}\,.
$$

As before, an agent has a guess $\alphat$ of how fast they think the increase will happen. Let us first consider someone without a trust fund and guess $\alphat$. For them, the factors to balance are:
\begin{align}
\frac{(1 - 1) \cdot s + T - s}{T} &= \frac{(1 - 1) \cdot s + T - s}{(1 - 1) \cdot s + \frac{\alphat}{2} (T-s)^2} \\
\frac{T - s}{T} &= \frac{T - s}{\frac{\alphat}{2} (T-s)^2} 
\end{align}

which gives switching point $s = T - \sqrt{\frac{2T}{\alphat}}\,$ and $\frac{T}{2} - \sqrt{\frac{T}{2\alphat}}$ time spent on the striving arm. However, an agent with a fixed large safety net (say, $R = T$) can explore for longer, trading off:

\begin{align}
    \frac{R - s + T - s}{R + T} &= \frac{R - s + T-s}{R - s + \frac{\alphat}{2} (T-s)^2} \\
    s &= T + \frac1\alphat - \sqrt{\frac{4T}{\alphat} + \frac 1\alphat}\,.
\end{align}

In this case, the switching point is as above, and the time spent on the striving arm is the same.

\end{document}